\documentclass[12pt]{article}

\usepackage{natbib}
\usepackage{amssymb,graphicx}
\usepackage{algorithmicx,algpseudocode}
\usepackage{algorithm}
\usepackage{amsthm}

\def\LM{{\mathrm{LM}}}
\def\LC{{\mathrm{LC}}}

\def \LM{{\rm LM}}
\def \LC{{\rm LC}}

\def\R{\mathbb{R}}
\def\F{\mathcal{F}}
\def\V{{\bf{V}}}
\def\idf{\mathcal{IF}}

\newtheorem{thm}{Theorem}[section]
\newtheorem{lem}[thm]{Lemma}
\newtheorem{rem}[thm]{Remark}
\newtheorem{defn}[thm]{Definition}
\newtheorem{prop}[thm]{Proposition}
\newtheorem{cor}[thm]{Corollary}

\newtheorem{exmp}[thm]{Example}

\begin{document}

\title{An Elimination Method to Solve\\ Interval Polynomial Systems}


\author{Sajjad Rahmany, Abdolali Basiri and Benyamin M.-Alizadeh\\ School of Mathematics and Computer Sciences,\\ Damghan University, Damghan, Iran.}

\maketitle
\begin{abstract}
There are several efficient methods to solve linear interval polynomial systems in the context of interval computations, however, the general case of interval polynomial systems is not yet covered as well. In this paper we introduce a new elimination method to solve and analyse interval polynomial systems, in general case. This method is based on computational algebraic geometry concepts such as polynomial ideals and Gr\"obner basis computation. Specially, we use the comprehensive G\"obner system concept to keep the dependencies between interval coefficients. At the end of paper, we will state some applications of our method to evaluate its performance.   
\end{abstract}


\section{Introduction}
Many computational problems arising from applied sciences deal with floating-point computation and so require to import polynomial equations containing error terms in computers. This redounds polynomial equations to appear with perturbed coefficients i.e the coefficients range in specific intervals and so these are called {\it interval polynomial equations}. Interval polynomial equations come naturally from several problems in engineering sciences such as control theory [35,36], dynamical systems [30] and so on. One of the most important problems in the context of interval polynomial equations is to analyse and study the stability and solutions of an (or a system of) interval polynomial(s).  More generally, the problem is to carry as much as possible information out from an interval polynomial system. Many scientific works are done in this direction using interval arithmetic [2], for instance computation of the roots in certain cases [6,7], however they do  not enable us to obtain the desired roots, at least approximately [6].
Another example consists those works which contain (the most popular) method to solve an interval polynomial equation by computing the roots of some exact algebraic polynomials, while it is hard to solve an algebraic equation of high degree which has its own complexity challenging problems [9,11,12,18,21]. There is also a new method described in [41] which counts the zeros of a univariate interval polynomial.
In addition to numerical methods, there are some attempts to combine numeric and symbolic methods to solve an interval polynomial system. In [8], Falai et al. state a modification of Wu's characteristic set method for interval polynomial systems, and use numerical approximation to find an interval containing the roots. 
The essential trick in this work is to omit all the terms with interval coefficients containing zero, which permits the division of interval coefficients simply. This consideration may fail some important polynomials.  

In this paper, we try to use exact symbolic methods to facilitate analysing interval polynomial systems thanks to algebraic elimination methods like  parametric computation techniques to analyse a parametric polynomial system which allows us to consider all exact polynomials arising from an interval polynomial. As we will state later, it is very important to keep the trace of interval coefficients during the computations. Roughly speaking we associate an auxiliary parameter to each interval coefficient provided that each parameter ranges over its own related interval only.  Nowadays there are important results [38,39], efficient algorithms [17,19,20,22,24,25] and powerful implementations [33,34] in the context of parametric computations and analysing parametric polynomial systems.
We introduce the new concept {\it interval Gr\"obner system} for a system of interval polynomials using the concept of {\it comprehensive Gr\"obner systems} [39] which is used to describe all different behaviours of a parametric polynomial system. Interval Gr\"obner system contains a finite number of systems where each one is a Gr\"obner basis for (non-interval) polynomial systems obtained from the main system. It is worth noting that against to [8], we don't omit any interval coefficient and cover all possible cases for the exact coefficients arising from the intervals. We also design an algorithm to compute interval Gr\"obner system of an interval polynomial system.
 
This paper is organized as follows. In Section 2 we state introductory definitions and recall interval arithmetic. 
In Section 3 we start to explain interval polynomials and their related concepts. We next receive to Section 4 which states the main idea behind this paper. To recall the concepts of computational algebraic tools we state Section 5 containing a brief introduction to Gr\"obner basis and comprehensive Gr\"obner system, together with their related algorithms. After, we describe our elimination method for interval polynomial systems in Section 6. Finally, as some applications and example, we state the Section 7 which contains two applied examples. 

\section{Preliminaries}
\label{pre}
In this section we recall the interval arithmetic and related concepts which are needed for the rest of this text. The main references of this section are [16] and [26]. Let $\R$ denote the set of real numbers while $\R^{*}$ is used to show the {\it extended real numbers set} i.e. $\R \cup \{-\infty,\infty\}$. 
\begin{defn}
Let $a,b\in \R^*$. We define $4$ kinds of {\it real intervals} defined by $a$ and $b$ as follows:
\begin{eqnarray}
\label{interval}
\begin{array}{rcl}
{\rm{Closed\ interval}} &:& [a,b] = \{x \mid a\leq x \leq b\} \ (a,b \ne \pm\infty) \\
{\rm{Left\ half\ open\ interval}} &:& (a,b] = \{x \mid a< x \leq b\}\ (b\ne \pm\infty)\\
{\rm{Right\ half\ open\ interval}} &:& [a,b) = \{x \mid a \leq x < b\}\ (a\ne \pm\infty)\\
{\rm{Open\ interval}} &:& (a,b) = \{x \mid a< x < b\} \\
\end{array}
\end{eqnarray} 
\end{defn}

It is worth noting that approximately all of existing texts in the subject of interval computation deal with closed intervals. 
Most of times we denote the intervals by capitals, and their lower (resp. upper) bounds by underbar (resp. overbar), as
$$X= [\underbar{X},\overline{X}]$$ 
to denote closed intervals. 
However, as there are some applied problems including non-closed intervals we preferred to consider all different types of intervals.

\begin{rem}
Having all different kinds of intervals at once, we use the notion $[a,b,i,j]$ where $i,j \in \{0,1\}$ to denote the intervals in (1), as follows:
\begin{eqnarray*}
[a,b,i,j] = \left\{
\begin{array}{lcl}
[a,b] & {\rm if} & i=j=1\\
(a,b] & {\rm if} & i=0, j=1\\

[a,b) & {\rm if} & i=1, j=0\\
(a,b)& {\rm if} & i=j=0\\
\end{array}
\right.
\end{eqnarray*} 
However when all of intervals come from one sort of presentation, we prefer to use the (1) form. 
\end{rem}
 Now we recall the interval arithmetic and discuss on the interval dependencies what will occur in the evaluation of interval expressions.  
 
There exist two equivalent ways to state interval arithmetic. The first is based on the endpoints of intervals while the second considers each interval as a subset of real numbers. 
Let $[a_1, b_1, i_1, j_1]$ and $[a_2, b_2, i_2, j_2]$ be two real intervals. Note that each real number $a$ is considered as $[a,a,1,1]$ which is called a degenerate interval. Four essential arithmetic operations are defined as follows:

\begin{eqnarray*}
[a_1, b_1, i_1, j_1] + [a_2, b_2, i_2, j_2] & = & [a_1 + a_2, b_1 + b_2, \min(i_1, i_2), \min(j_1, j_2)]\\ ~ 
 [a_1, b_1, i_1, j_1] - [a_2, b_2, i_2, j_2] & = & [a_1 -b_2, b_1 -a_2, \min(i_1, j_2), \min(j_1, i_2)]\\ ~
 [a_1, b_1, i_1, j_1] \times [a_2, b_2, i_2, j_2] & = & [a_kb_\ell, a_{k'}b_{\ell'}, \min(i_k,j_\ell), \min(i_{k'},j_{\ell'})]\\ ~
\end{eqnarray*}
where $a_kb_\ell$ and $a_{k'}b_{\ell'}$ are the minimum and maximum of the set $\{a_1a_2, a_1b_2, b_1a_2, b_1b_2\}$ respectively, and finally
\begin{eqnarray*} 
[a_1, b_1, i_1, j_1] / [a_2, b_2, i_2, j_2] & = & [a_1, b_1, i_1, j_1] \times [1/b_2, 1/a_2, j_2, i_2]
\end{eqnarray*}
provided that $a_2> 0$ or $b_2< 0$ or $a_2 = i_2 = 0$ or $b_2 = j_2 = 0$. 

Note in the above relations that all ambiguous cases $\infty-\infty$, $\pm\infty\times 0$, $\frac{\pm\infty}{\pm\infty}$ and $\frac{0}{0}$ will induce the biggest possible interval i.e. $\R$.

As an easy observation, when an interval $X=[a,b,i,j]$ with $a\ne b$ contains zero,  we can compute $1/X$ as follows:
\begin{itemize}
\item If $a=0$ then 
$$\frac{1}{X} = \frac{1}{[a,b,0,j]}=[\frac{1}{b},+\infty,j,0],$$
\item If $b=0$ then
$$\frac{1}{X} = \frac{1}{[a,b,i,0]}=[-\infty, \frac{1}{a},0,i],$$
\item If $ab<0$ then by seperating $X$ as 
$X=[a,0,i,0] \cup [0,b,0,j]$ we have 
$$\frac{1}{X} = \frac{1}{[a,0,i,0] \cup [0,b,0,j]} = [-\infty, \frac{1}{a},0,i] \cup [\frac{1}{b}, +\infty, j , 0].$$
\end{itemize}
Now consider $A$ and $B$ are two intervals as two sets of real numbers. We can state the above definitions of four essential arithmetic operations as 
$$ A\ op \ B = \{x\ | \ \exists\ a\in A, \ b\in B: \ x=a \ op \ b\}$$
where $op \in \{+,-,\times, /\}$.  
\begin{rem}
Although interval arithmetic seems to be compatible with real numbers arithmetic, but this affects distributivity of multiplication over addition and the existance of inverse elements. More preciesly
for each intervals $X,Y$ and $Z$,
\begin{itemize}
\item[$\bullet$] $X\times (Y+Z)\subseteq X\times Y+X\times Z$,

\noindent
and if $X$ is non-degenerated then 
\item[$\bullet$] $X\times \frac{1}{X} \ne 1$, but $1\in X\times \frac{1}{X}$ and 
\item[$\bullet$] $X+(-X) \ne 0$, but $0\in X+(-X)$. 
\end{itemize}
Furthermore, if $X$ contains some negative real numbers, then
$$X^n\ne \underbrace{X\cdots X}_{n \ times}.$$
To see this, let for instance $X=[a,b,i,j]$ where $a<0$ and $|a|<b$. Then, $X^2 = [0,b^2,1,j]$ while $X\times X = [ab,b^2, \min(i,j),j]$. To solve this inconsistency, we define the $n^,$th power of an interval for non-negative integer $n$ as follows:
\begin{eqnarray*}
[a,b,i,j]^n=\left\{
\begin{array}{lcl}
1 &  & n=0\\

[a^n, b^n, i, j] &  & 0\leq a\\

[b^n, a^n, j, i] &  & b\leq 0 \\

[0, \max(a^n,b^n), 1, c] &  & a< 0 < b  \\
\end{array}
\right.
\end{eqnarray*}
where 
$ c=\left\{
\begin{array}{crl}
i& & |b|<|a|\\
j& & otherwise
\end{array}
\right..
$
\end{rem}
Let us now evaluate some expressions to illustrate more challenging problems dealing with interval arithmetic. Let $f(x,y)=\frac{x}{x+y}$, $X=[1,2]$ and $Y=[1,3]$. We compute $f(X,Y)$ in two ways. The first way is to compute $f(X,Y)$ as an usual evaluation using interval arithmetic:
\begin{eqnarray}
\label{eq:2}
\frac{X}{X+Y} = [1,2]/[2,5] = [1/5, 1].
\end{eqnarray}
However, one can manipulate the expression to see
\begin{eqnarray}
\label{eq:3}
\frac{X}{X+Y} = \frac{1}{1+\frac{Y}{X}} = 1 /[3/2, 4] = [1/4 , 2/3]
\end{eqnarray}
 Let us separate $x$ and $y$ first depending on $f(x,y)$: we call $y$ (resp. $x$), a {\it first} (resp. {\it second}) {\it class} variable of $f$ as it appears one (resp. more than one) time in the structure of $f$. As it can be easily seen, the answer of (3) is a narrower interval and in fact the exact value. The reason is that $X$ is a second class variable for (2) and so it brings {\it dependency} between two parts of the expression. This is while there is no dependency in (3) given that $\frac{Y}{X}$ appears only one time, and so it is a first class variable.   Dependency is one of the crucial points of this paper to find the solution set of interval polynomial systems.

Dependencies are the main reason that causes to appear an amount of error by introducing larger intervals than the exact solution as all publications in this area try to avoid dependencies. Nevertheless it is possible to cancel dependencies by considering $X-X=0$ (see [16]) as well as $X/X=1$ easily, while sometimes this becomes a difficult work.

\section{Interval Polynomials}
\label{IntPoly}
Let $\R$ be the field of real numbers, considered as the ground field of computations all over the current text and consider the set $\{x_1,\ldots,x_n\}$  to be the set of variables.
\begin{defn}
Each polynomial of the form 
\begin{eqnarray}
\label{Def:1}
[f]=\sum_{i=1}^m [a_i,b_i,\ell_i, k_i] x_1^{\alpha_{i1}}\cdots x_n^{\alpha_{in}}
\end{eqnarray}
is called an interval polynomial, where $[a_i,b_i,\ell_i,k_i]$ is a real interval for each $i=1,\ldots,m$, and each power product $x_1^{\alpha_{i1}}\cdots x_n^{\alpha_{in}}$ is called a monomial where the powers are non negative integers. We denote the set of all interval polynomials by $[\R][x_1,\ldots,x_n]$.
\end{defn}
\begin{defn}
Let $[f]$ be an interval polynomial as defined in (4). The set of all polynomials arising from $[f]$ for different values of intervals in coefficients is called the family of $[f]$ and is denoted by $\mathcal{F}([f])$.
More preciesly:
$$\mathcal{F}([f]) = \{\sum_{i=1}^m c_i x_1^{\alpha_{i1}}\cdots x_n^{\alpha_{in}} \ | \ c_i \in [a_i,b_i,\ell_i, k_i], i=1,\ldots,m\}$$ 
\end{defn}
Similar to the family of an interval polynomial, we can define the family of a set of interval polynomials as follows:
\begin{defn}
\label{Def:2}
Let $S:=\{[f]_1,\ldots,[f]_\ell\}$ be a set of interval polynomials with $\mathcal{F}([f]_j) = \mathcal{F}_j$ for each $j=1,\ldots, \ell$. We define the family of $S$ to be the set $\mathcal{F}_1\times\cdots \times \mathcal{F}_\ell$, denoted by $\mathcal{F}(S)$. 
\end{defn}
 We now define the concept of solution set for an interval polynomial.
\begin{defn}
\label{Def:3}
For an interval polynomial $[f] \in {[\R][x_1,\ldots,x_n]}$, we say that $a=(a_1,\ldots,a_n) \in \R^n$ is a real solution or a real root of $[f]$, if there exists a polynomial $p \in \mathcal{F}([f])$ such that $p(a) = 0$.
Similarily, we say that a system of interval polynomials has a solution, if the contained interval polynomials  have a common solution.
\end{defn}
\begin{exmp}
Let us find the solution set of 
$$[-2,-1] x^2 + [1,2] x + [1,3]=0$$ 
where all intervals are closed. 
When $x\geq 0$, we have 
$$[-2,-1] x^2 + [1,2] x + [1,3] = [-2x^2+x+1, -x^2+2x+3] =[0,0]$$
So we must have 
$$0\leq -2x^2+x+1\ {\rm{and}}\ -x^2+2x+3\leq 0$$
which implies
$$x\in[1,3].$$
Similarly when $x\leq 0$, we have
$$[-2,-1] x^2 + [1,2] x + [1,3] = [-2x^2+2x+1, -x^2+x+3] =[0,0]$$
or equivalently 
$$0\leq -2x^2+2x+1\ {\rm{and}}\ -x^2+x+3\leq 0$$
which concludes 
$$x\in [-1.31, -0.36].$$
Thus the solution set of this interval polynomial is 
$$ [-1.31, -0.36] \cup [1,3].$$
It is notable that using interval arithmetic in the well-known solution way of a quadratic polynomial equation due to discriminant, we receive to 
$$[-2.15, -0.06]\cup [0.382 , 8.4]$$ 
as the solution set that contains an amount of error. 
\end{exmp}

\section{The idea}
\label{Idea}
In this section we are going to describe the problems which may occur using the usual elimination method on a system of interval polynomials. To facilitate the description, let us give an example. Consider the system containing 
$$f_1=[1,2]x_1+x_2+2x_3,\ f_2=[1,4]x_1+x_2+1,\ f_3=[3,4]x_1+x_2+4x_3.$$
Going to eliminate the variable $x_2$, we have
$$f_1-f_2=[-3,1]x_1+2x_3-1,\ f_3-f_2=[-1,3]x_1+4x_3-1.$$
One may now conclude that if we choose $0$ from both intervals $[-3,1]$ and $[-1,3]$ then the system has no solution because $2x_3-1=4x_3-1=0.$ However this case is impossible since both of intervals $[-3,1]$ and $[-1,3]$ can not be zero at once! To see this, notice that $[-3,1]$ comes from $[1,2]-[1,4]$ and so for $[-3,1]$ being zero, $[1,4]$ must give some values in $[1,2]$. On the other hand, 
$[-1,3]$ comes from $[3,4]-[1,4]$ and so this interval can be zero only when $[1,4]$ gives some values in $[3,4]$ and this is a contradiction. This simple linear system shows that the usual elimination method can cause a wrong decision or appearing some extra values in the solution set. The main reason is that we forgot the dependencies between $[-3,1]$ and $[-1,3]$ during the computation while they are both dependent on $[1,4]$ and so they are dependent.

To solve this problem, we must keep the trace of each interval coefficient. In doing so, our idea is to use a parameter instead of each interval, to see how new coefficients are built. For instance let us substitute $[1,2]$, $[1,4]$ and $[3,4]$ by $a$, $b$ and $c$ as parameters in the above example. So we have  
$$\tilde{f_1}=ax_1+x_2+2x_3,\ \tilde{f_2}=bx_1+x_2+1,\ \tilde{f_3}=cx_1+x_2+4x_3$$
and doing elimination steps we have
$$\tilde{f_1}-\tilde{f_2}=(a-b)x_1+2x_3-1,\ \tilde{f_3}-\tilde{f_2}=(c-b)x_1+4x_3-1.$$
Now one can conclude that under the assumption that $a-b=0$, the coefficient $c-b$ can not be zero. The reason is that if $c-b=0$ then $a=c$ while $1\leq a \leq 2$ and $3\leq c \leq 4$. 
Therefore using parameters prevent us taking wrong decisions about the solution set. 

The main question here is that how we can use elimination method when the coefficients contain some parameters. In fact as we will state in the next sections, we do the elimination steps thanks to Gr\"obner basis and for the parametric case, we use the concept of comprehensive Gr\"obner system, to see the simplest possible polynomials to solve.  So our idea is to convert the interval polynomial systems to a parametric polynomial system and use the parametric algorithmic aspects, of course with some modifications, to solve the parametric system by dividing the solution set into finitely many components. At the end, we convert the result to see the solution set of the interval polynomial system. 

\section{Gr\"obner Bases and Comprehensive Gr\"obner Systems}
\label{GB}
In this section we recall the concepts and notations of ordinary and parametric polynomial rings. 
Let $K$ be a  field and $ x_1,\ldots,x_n$ be $n$ (algebraically independent) variables. Each power product $x_1^{\alpha_1}\cdots x_n^{\alpha_n}$ is called a  monomial where $\alpha_1,\ldots,\alpha_n\in \mathbf{Z}_{\geq 0}$. Because of simplicity, we abbreviate such monomials by ${\bf x}^\alpha$ where ${\bf x}$ is used for the sequence $x_1,\ldots,x_n$ and $\alpha=(\alpha_1,\ldots,\alpha_n)$. We can sort the set of all monomials over $K$ by special types of total orderings so called monomial orderings, recalled in the following definition.
\begin{defn}
The total ordering $\prec$ on the set of monomials is called a monomial ordering whenever for each monomials ${\bf x}^\alpha, {\bf x}^\beta$ and ${\bf x}^\gamma$ we have:
\begin{itemize}
\item
${\bf x}^\alpha\prec {\bf x}^\beta \Rightarrow {\bf x}^\gamma {\bf x}^\alpha\prec {\bf x}^\gamma {\bf x}^\beta$, and
\item
$\prec$ is well-ordering.
\end{itemize}
\end{defn} 
There are infinitely many monomial orderings, each one is convenient for a special type of problems. Among them, we point to pure and graded reverse lexicographic orderings denoted by $\prec_{lex}$ and $\prec_{grevlex}$ as follows. assume that $x_n\prec \cdots \prec x_1$. We say that
\begin{itemize}
\item ${\bf x}^\alpha \prec_{lex} {\bf x}^\beta$ whenever 
$$\alpha_1=\beta_1,\ldots,\alpha_i=\beta_i \ {{\rm and}} \ \alpha_{i+1}<\beta_{i+1}$$
for an integer $1\leq i < n$. 
\item ${\bf x}^\alpha \prec_{grevlex} {\bf x}^\beta$ if
$$\sum_{i=1}^n \alpha_i < \sum_{i=1}^n \beta_i $$
breaking ties when there exists an integer $1\leq i < n$ such that
$$\alpha_n=\beta_n,\ldots,\alpha_{n-i}=\beta_{n-i} \ { {\rm and}} \ \alpha_{n-i-1}>\beta_{n-i-1}.$$
\end{itemize} 
It is worth noting that the former has many theoretical importance while the latter speeds up the computations and carries fewer information out. 

Each $K-$linear combination of monomials is called a polynomial on $x_1,\ldots,x_n$ over $K$. The set of all polynomials has the ring structure with usual polynomial addition and multiplication, and is called the polynomial ring on $x_1,\ldots,x_n$ over $K$ and denoted by $K[x_1,\ldots,x_n]$ or just by $K[{\bf x}]$. Let $f$ be a polynomial and $\prec$ be a monomial ordering. The greatest monomial w.r.t. $\prec$ contained in $f$ is called the leading monomial of $f$, denoted by $\LM(f)$ and the coefficient of $\LM(f)$ is called the leading coefficient of $f$ which is pointed by $\LC(f)$. Further, if $F$ is a set of polynomials, $\LM(F)$ is defined to be $\{\LM(f) | f\in F\}$ and if $I$ is an ideal, $in(I)$ is the ideal generated by $\LM(I)$ and is called the initial ideal of $I$.  We are now going to remind the concept of Gr\"obner basis of a polynomial ideal which carries lots of useful information out about the ideal.  
\begin{defn}
Let $I$ be a polynomial ideal of $K[{\bf x}]$ and $\prec$ be a monomial ordering. The finite set $G\subset I$ is called a Gr\"obner basis of $I$ if for each non zero polynomial $f\in I$, $\LM(f)$ is divisible by $\LM(g)$ for some $g\in G$. 
\end{defn}
Using the well-known Hilbert basis theorem (See [4] for example), it is proved that each polynomial ideal possesses a Gr\"obner basis with respect to each monomial ordering. There are efficient algorithms also to compute Gr\"obner basis. The first and the most simplest one is the Buchberger algorithm which is devoted in the same time of introduction of Gr\"obner basis concept  while he most efficient known algorithm is the Faug\`ere's F$_5$ algorithm [10] and another signature-based algorithms such as G$^2$V [13] and GVW [14].
It is worth noting that Gr\"obner basis of an ideal is not unique necessarily.  To have unicity, we define the reduced Gr\"obner basis concept.
As an important fact  the reduced Gr\"obner basis of an ideal is unique up to the monomial ordering.
\begin{defn}
Let $G$ be a Gr\"obner basis for the ideal $I$ w.r.t. $\prec$. Then $G$ is called a reduced Gr\"obner basis of $I$ whenever each $g\in G$ is monic, i.e. $\LC(g)=1$ and none of the monomials appearing in $g$ is divisible by $\LM(h)$ for each $h\in G\setminus \{g\}$.
\end{defn}
One of the most important applications of Gr\"obner basis is its help to solve a polynomial system. Let 
$$
\left\{
\begin{array}{ccc}
f_1&=&0\\
&\vdots&\\
f_k&=&0\\
\end{array}
\right.
$$
be a polynomial system and $I=\langle f_1,\ldots,f_k\rangle$ be the ideal generated by $f_1,\ldots,f_k$. We define the affine variety associated to the above system or equivalently to the ideal $I$ to be
$${\bf V}(I) = {\bf V}(f_1,\ldots,f_k) = \{\alpha \in \overline{K}^n | f_1(\alpha)=\cdots = f_k(\alpha) = 0\}$$
where $\overline{K}$ is used to denote the algebraic closure of $K$.
Now let $G$ be a Gr\"obner basis for $I$ with respect to an arbitrary monomial ordering. As an interesting fact, $I=\langle G\rangle$ which implies that ${\bf V}(I) = {\bf V}(G)$. This is the key computational trick to solve a polynomial system. Let us continue by an example.
\begin{exmp}
\label{Ex:1}
We are going to solve the following polynomial system:
$$
\left\{
\begin{array}{ccc}
x^2-xyz+1&=&0\\
y^3+z^2-1&=&\\
xy^2+z^2&=&0\\
\end{array}
\right.
$$
By the nice properties of pure lexicographical ordering, the reduced Gr\"obner basis of the ideal $I=\langle x^2-xyz+1, y^3+z^2-1, xy^2+z^2 \rangle \subset \mathbf{Q}[x,y,z]$ has the form
$$G=\{g_1(z), x-g_2(z), y-g_3(z)\}$$
 w.r.t. $z\prec_{lex}y\prec_{lex}x$, where
 {\small
 $$
 \left[
 \begin{array}{ccl}
 g_1(z)&=& z^{15}-3z^{14}+5z^{12}-3z^{10}-z^9-z^8+4z^6-6z^4+4z^2-1\\
 g_2(z)&=& 2z^{14}-9z^{13}+11z^{12}+2z^{11}-7z^{10}-3z^9+2z^8-z^7+4z^6+\\
 &&+7z^5-10z^4-6z^3+11z^2+2z-4\\
 g_3(z)&=&z^{13}-3z^{12}+z^{11}+2z^{10}+z^9-z^8-2z^6+2z^4-z^3-3z^2+1.\\
 \end{array}
 \right.
  $$ 
}
This special form of Gr\"obner basis for this system allows us to find ${\bf V}(G)$ by solving only one univariate polynomial $g_1(z)$ and putting the roots into the two last polynomials in $G$.
\end{exmp}

Suppose now that the same system of Example 5.4 is given as follows with parametric coefficients on parameters are $a,b$ and $c$:
$$
\left\{
\begin{array}{ccc}
{a}x^2-({a}^2-{b}+1)xyz+1&=&0\\
y^3+{c}^2z^2-1&=&\\
({a}+{b}+{c})xy^2+z^2&=&0\\
\end{array}
\right.
$$
The solutions of this system depend on the values of parameters apparently as we can see that the system has no solutions whenever ${a}=0$ and ${b}=1$ while it converts to the system of Example 5.4 for ${a}=1, {b}=1$ and ${c}=-1$ and so has some solutions. To manage all of different behaviours of parameters which cause to different behaviour of the main system, we recall the concept of comprehensive Gr\"obner system in the sequel. By this we can divide the space of parameters, i.e. $\overline{K}^t$  into a finite number of partitions, for which the general form of polynomials contained in assigned Gr\"obner basis is known.
 
 Let $K$ be a field and ${\bf a} := a_1,\ldots,a_t$ and ${\bf x}:=x_1,\ldots,x_n$ be the sequences of parameters and variables respectively. We call $K[{\bf a}][{\bf x}]$, the parametric polynomial ring over $K$, with parameters ${\bf a}$ and variables ${\bf x}$. This ring is in fact the set of all parametric polynomials as 
$$\sum_{i=1}^{m} p_{i} {\bf x}^{\alpha_i}$$
where $p_{i}\in K[{\bf a}]$ is a polynomial on ${\bf a}$ with coefficients in $K$, for each $i$. 

\begin{defn}
\label{Def:CGS}
Let $I\subset K[{\bf a}][{\bf x}]$ be a parametric ideal and $\prec$ be a monomial ordering on ${\bf x}$. Then the set 
$$\mathcal{G}(I)=\{(E_i,N_i,G_i) \mid i=1,\ldots,\ell\} \subset K[{\bf a}]\times K[{\bf a}]\times K[{\bf a}][{\bf x}]$$
is said a comprehensive  Gr\"obner system for $I$ if for each $(\lambda_1,\ldots,\lambda_t)\in \overline{K}^t$ and each specialization
\begin{eqnarray*}
\sigma_{(\lambda_1,\ldots,\lambda_t)} :& K[{\bf a}][{\bf x}] &\rightarrow  \ \ \ \ \ \ \ \overline{K}[{\bf x}]\\
 &\sum_{i=1}^{m} p_{i} {\bf x}^{\alpha_i} &\mapsto \sum_{i=1}^{m} p_{i}(\lambda_1,\ldots,\lambda_t) {\bf x}^{\alpha_i}
\end{eqnarray*}
there exists an $1\leq i \leq \ell$ such  that $(\lambda_1,\ldots,\lambda_t) \in {\bf V}(E_i) \setminus {\bf V}(N_i)$ and $\sigma_{(\lambda_1,\ldots,\lambda_t)}(G_i)$ is a Gr\"obner basis for $\sigma_{(\lambda_1,\ldots,\lambda_t)}(I)$ with respect to $\prec$. Because of simplicity, we call $E_i$ and $N_i$ the null and non-null  conditions respectively.
\end{defn}
Remark that, by [39, Theorem 2.7], every parametric ideal has a comprehensive Gr\"obner system.
Now we give an example from [20] to illustrate the definition of comprehensive Gr\"obner system.
\begin{exmp}
Consider the following parametric polynomial system in $\mathbf{Q}[a,b,c][x,y]$:
\begin{eqnarray*}
\Sigma:\left\{
\begin{array}{lll}
ax-b&=&0\\
by-a&=&0\\
cx^2-y&=&0\\
cy^2-x&=&0
\end{array}
\right.
\end{eqnarray*}
Choosing the graded reverse lexicographical ordering $y\prec x$, we have the following comprehensive Gr\"obner system:
\begin{figure}[H]
 \centering {\small
 \begin{tabular}{|c||c||c|}
 \cline{1-3}
$G_i$&$E_i$&$N_i$\\
 \cline{1-3}
 $\{1\}$&$\{\  \}$&$\{a^6-b^6, a^3c-b^3, b^3c-a^3,$\\
 && $ac^2-a, bc^2-b\}$\\
 \cline{1-3}
$ \{bx-acy, by-a\}$&$\{a^6-b^6, a^3c-b^3, b^3c-a^3,$&$\{b\}$\\
& $ac^2-a, bc^2-b\}$&\\
\cline{1-3}
 $\{cx^2-y, cy^2-x\}$&$\{a, b \}$&$\{c\}$\\
\cline{1-3}
$\{x,y\}$& $\{a,b,c\}$ & $\{\  \}$\\
\cline{1-3}
 \end{tabular}
 }
 \end{figure}
For instance, for the specialization $\sigma_{(1,1,1)}$ for which $a\mapsto 1, b\mapsto 1$ and $c\mapsto 1$, 
$$\sigma_{(1,1,1)}(\{bx-acy, by-a\})=\{x-y,y-1\}$$ 
is a Gr\"obner basis of $\sigma_{(1,1,1)}(\langle \Sigma \rangle)$. 
\end{exmp}

It is worth noting that if ${\bf V}(E_i) \setminus {\bf V}(N_i) = \emptyset$ for some $i$, then the triple $(E_i,N_i,G_i)$ is useless, and so it must be omitted from the comprehensive Gr\"obner system. In this case we say that the pair $(E_i,N_i)$ is {\it inconsistent}. It is easy to see that inconsistency occurs if and only if $N_i \subset \sqrt{\langle E_i \rangle}$ and so we need to an efficient radical membership test to determine inconsistencies. In [19, 20] there is  a new and efficient algorithm to compute comprehensive Gr\"obner system of a parametric polynomial ideal which uses a new and powerful radical membership criterion. Therefore we prefer to employ this algorithm so called {\sc PGB} algorithm in our computations. 
Another essential trick which is used in [20] is the usage of {\it minimal Dickson basis} which reduces the content of computations in {\sc PGB}. Before explain it, 
let us recall some notations which are used in the structure of {\sc PGB}. Let $\prec_{\bf x}$ and $\prec_{\bf a}$ be two monomial orderings on $K[\bf x]$ and $K[\bf a]$ respectively. Let also $\prec_{{\bf x,a}}$ be the block ordering of $\prec_{\bf x}$ and $\prec_{\bf a}$, comparing two parametric monomials by $\prec_{\bf x}$, breaking tie by $\prec_{\bf a}$. For a parametric polynomial $f \in K[\bf a][\bf x]$, we denote by $\LM_{\bf x}(f)$ (resp. by $\LC_{\bf x}(f)$) the leading monomial (resp. the leading coefficient) of $f$ when it is considered as a polynomial in $K({\bf a})[{\bf x}]$, and so $\LC_{\bf x}(f) \in K[\bf a]$. 

\begin{defn}
\label{Def:MDB}
By the above notations, let $P\subset K[{\bf a}][{\bf x}]$ be a set of parametric polynomials and $G\subset P$. Then, $G$ is called a minimal Dickson basis of $P$ denoted by ${\sc MDBasis}(P)$, if: 
\begin{itemize}
\item
For each $p\in P$, there exist some $g\in G$ such that $\LM_{\bf x}(g) \mid \LM_{\bf x}(p)$ and
\item
For each two distinct polynomials in $G$ as $g_1$ and $g_2$, none of $\LM_{\bf x}(g_1)$ and $\LM_{\bf x}(g_2)$ divides another. 
\end{itemize}
\end{defn}
The case which occurs in {\sc PGB} to compute a minimal Dickson basis for $P$ is only when $P$ is a Gr\"obner basis for $\langle P \rangle$ itself w.r.t. $\prec_{{\bf x}, {\bf a}}$ and $P \cap K[{\bf a}]=\{0\}$. In this situation, it suffices by Definition 5.7 to omitt all polynomials $p$ from $P$ for which there exists a $p'\in P$ such that $\LM_{\bf x}(p') \mid \LM_{\bf x}(p)$. 

The {\sc PGB} algorithm as is shown below, uses {\sc PGB-main} algorithm to introduce new branches in computations. 
\begin{algorithm}[H]
\caption{{\sc PGB}}
\begin{algorithmic}[1]
\Procedure {\sc PGB}{$P, \prec_{\bf a}, \prec_{\bf x}$} 
    \State {$E,N:=\{\ \},\{1\};$}
    \State {$\prec_{{\bf x}, {\bf a}}:=$The block ordering of $\prec_{\bf x}, \prec_{\bf a}$}
    \State {\bf Return} {\sc PGB-main$(P,E,N,\prec_{{\bf x}, {\bf a}})$};
\EndProcedure
\end{algorithmic}
\end{algorithm}

The main work of {\sc PGB-main} is to create all necessary branches and import them in comprehensive Gr\"obner system at output. 
In this algorithm $A*B$ is defined to be the set $\{ab \mid a\in A, b\in B \}$.
\begin{algorithm}[H]
\caption{{\sc PGB-main}}
\begin{algorithmic}
\Procedure {{\sc PGB-main}}{$P, E, N, \prec_{{\bf x}, {\bf a}}$}
    
    \State {$G:=$ The reduced Gr\"obner basis for $P$ w.r.t. $\prec_{{\bf a}, {\bf x}}$}
    \If {$1 \in G$}
    	\State {\bf Return} $(E,N,\{1\})$;
    \EndIf
    \State {$G_r:=G \cap K[\bf a];$}
    \If {{\sc IsConsistent}$(E,N*G_r)$}
    	\State {$PGB:= \{(E,N*G_r,\{1\})\}$};
    \Else
    	\State{$PGB:=\emptyset$};
    \EndIf
    \If {{\sc IsConsistent}$(G_r,N)$}
    	\State {$G_m:=${\sc MDBasis}$(G \setminus G_r);$}
    \Else
    	 \State {\bf Return} $(PGB)$;
    \EndIf
    \State{$h:={\rm lcm}(h_1,\ldots,h_k)$, where $h_i = \LC_{\bf x}(g_i)$ and $g_i \in G_m$};
    \If{{{\sc IsConsistent}$(G_r,N*\{h\})$}}
    		\State {$PGB:=PGB \cup \{(G_r,N*\{h\},G_m)\}$;}
    \EndIf
    \For{$i=1,\ldots,k$}
    	\State {$PGB:=PGB  \ \cup $ {\sc PGB-main}$(G\setminus G_r, G_r \cup \{h_i\},N*\{\prod_{j=1}^{i-1} h_j\},\prec_{{\bf a}, {\bf x}})$}
     \EndFor
\EndProcedure
\end{algorithmic}
\end{algorithm}
As it is shown in the algorithm, it computes first a Gr\"obner basis of the ideal $\langle P \rangle$ over $K[\bf{a,x}]$ i.e. $G$, before performing any branches based on parametric constraints, according to [20, Lemma 32] as follows:
\begin{lem}
By the notations used in the algorithm, for each specialization $\sigma_{(\lambda_1,\ldots,\lambda_t)}$ if
$$(\lambda_1,\ldots,\lambda_t) \in {\bf V}(G_r) \setminus {\bf V}(\prod_{g\in G\setminus G_r}\LC_{\bf x}(g))$$ 
then $\sigma_{(\lambda_1,\ldots,\lambda_t)}(G)$ is a Gr\"obner basis for $\sigma_{(\lambda_1,\ldots,\lambda_t)}(\langle P \rangle)$.
\end{lem}
After this, the algorithm computes a minimal Dickson basis i.e. $G_m$ and continues by taking a decision for each situation that one of the leading coeffiecients of $G_m$ is zero. By this, {\sc PGB-main} constructs all necessary branches to import in comprehensive Gr\"obner system. All over the algorithm, when it needs to add a new branch $(E_i,N_i,G_i)$ into the system, the algorithm {\sc IsConsistent} is used as follow to test the consistency of parametric conditions $(E_i,N_i)$.

\begin{algorithm}[H]
\caption{{\sc IsConsistent}}
\begin{algorithmic}
\Procedure {{\sc IsConsistent}} {$E,N$}
    \State {$flag:=$false;}
    \For {$g\in N$ {\bf while} $flag=$false}
    	\If{$g\notin \sqrt{\langle E\rangle}$}
    		\State{$flag:=$true;}
    	\EndIf
    \EndFor
    \State{{\bf Return} $flag$};
\EndProcedure
\end{algorithmic}
\end{algorithm}
The main part of this algorithm is radical membership test. The powerful trick which is used in [19,20] to radical membership check is based on linear algebra methods tackling  with a probabilistic check. We refer the reader to [20, Section 5] for more details.

\section{Elimination Method to Solve Interval Polynomial Systems}
\label{Elim}
In this section we introduce the new concept of {\it interval Gr\"obner system} and its related definitions and statements. 

 Now we state the following proposition as an immediate consequence of Definition 3.4. Recall that for  a polynomial  system $S \subset \R[x_1,\ldots,x_n]$ the variety of $S$ is the set of all complex solutions of $S$, denoted by $\V(S)$.
\begin{prop}
A system $ [S]$ of interval polynomials has a solution if and only if there exists a polynomial system $S$ in $\F( [S])$ with $\V(S)\neq \emptyset$.  
\end{prop}

There is an efficient criterion due to the well-known Hilbert Nullestelensatz theorem which determines if  $\V(S) \neq \emptyset$  by Gr\"obner basis: $\V(S) \neq \emptyset$ if and only if the Gr\"obner basis of $\langle S \rangle$ does not contain any constant. Note that there are infinitely many polynomial systems in $\F ( [S])$ for an interval polynomial system $ [S]$ and so it is practical impossible to check  all of them by Nullestelensatz theorem. Nevertheless, we give a finite partition on the set of all polynomial systems arising from $ [S]$ using the concept of comprehensive Gr\"obner system.
\begin{defn}
Let $ [S]=\{[f]_1,\ldots,[f]_\ell\}$ be a system of interval polynomials. We define the ideal family of $ [S]$, denoted by $\idf( [S])$ to be the set 
$$\idf( [S]) = \{\langle p_1,\ldots,p_\ell \rangle \ | \ (p_1,\ldots,p_\ell) \in \F( [S])\}$$
\end{defn}

 \begin{thm}
 \label{Thm:main}
Let $ [S]$ be a system of interval polynomials and $\prec$ be a monomial ordering on $\R[x_1,\ldots,x_n]$. Then 
\begin{itemize}
\item The set of initial ideals $\{in(I) \ | \ I \in \idf( [S])\}$ is a finite set, and
\item For each set of ideals of $\idf( [S])$ with the same initial ideal, there exists a set of parametric polynomials which induces the ideals by different specializations.
\end{itemize}
 \end{thm}
\begin{proof}
 To prove this theorem, we use the concept of comprehensive Gr\"obner system. Suppose that $S^* $ is obtained by replacing each interval coefficient by a parameter. Note that  if an interval appears in $t\geq 1$ coefficients, then we assign $t$ distinct parameters to it. It is easy to check that each element of $\idf( [S])$ is the image of $S^*$ under a suitable specialization. On the other hand by [39, Theorem 2.7] $S^*$ has a finite comprehensive Gr\"obner system as $\mathcal{G} = \{(E_1,N_1,G_1), \ldots, (E_k,N_k,G_k)\}$, where for each specialization $\sigma$ there exists a $1\leq j \leq k$ such that $\LM(\sigma(S^*)) = \LM(G_i)$. It is worth noting that although there is a finite number of branches in $\mathcal{G}$, we can also remove those specializations with complex values, and also those with values out of the assigned interval. Thus for each $I \in \idf( [S])$ there exists an $1\leq i \leq k$ with $in(I) = \langle\LM(G_i)\rangle$ and this finishes the proof.
 \end{proof}
What is explained in the proof of Theorem 6.3 yields to extend the concept of comprehensive Gr\"obner system for interval polynomials. 
\begin{defn}
\label{Def:IGS}
Let $ [S]\subset [\R][x_1,\ldots,x_n]$ be a system of interval polynomials with $t$ interval coefficients, and $\prec$ be a monomial ordering on $\R[x_1,\ldots,x_n]$. Let also that $\mathcal{G} = \{(E_1,N_1,G_1), \ldots, (E_k,N_k,G_k)\}$ be a set of triples 
$$(E_i,N_i,G_i) \in \R[h_1,\ldots,h_t] \times \R[h_1,\ldots,h_t]\times \R[h_1,\ldots,h_t][x_1,\ldots,x_n] $$
where $\{h_1,\ldots,h_t\}$ is the set of parameters assigned to each interval coefficient. Then we call $\mathcal{G}$ an interval Gr\"obner system for $ [S]$ denoted by $\mathcal{G}_{\prec}( [S])$ if for each $t-$tuple $(a_1,\ldots,a_t)$ of the inner values of interval coefficients there exists an $1\leq i \leq k$ such that:
\begin{itemize}
\item 
For each  $p \in E_i$, $p(a_1,\ldots,a_t) = 0$,
\item 
There exist some $q \in N_i$ such that $q(a_1,\ldots,a_t) \neq0,$ and
\item
$\sigma(G_i)$ is a Gr\"obner basis for $\langle \sigma( [S]) \rangle$ with respect to $\prec$, where $\sigma$ is the specialization $h_j \mapsto a_j$ for $j=1,\ldots,t$.
\end{itemize} 
\end{defn}
\begin{thm}
Each interval polynomial system possesses an interval Gr\"obner system.
\end{thm} 
\begin{proof}
Let $S^*$ be the parametric polynomial system obtained by assigning each interval coefficient to a parameter. As mentioned in the proof of Theorem 6.3, $\mathcal{G}_{\prec}( [S])$ is the same comprehensive Gr\"obner system of $S^*$ where each parameter is bounded to give values from its assigned ideal. On the other hand it is proved that each system of parametric polynomials has a comprehensive Gr\"obner system, which terminates the proof.
\end{proof}
We give now an easy example to illustrate what described above.
\begin{exmp}
\label{IGSexample}
Consider the interval polynomial system 
\begin{eqnarray}
\label{IGSexamplesystem} 
[S] = \left\{
\begin{array}{ccc}
[-1, 2)xy+[0, 1)y+[3, 5)&=&0\\
~ [-3, 1)xy^2+[1, 3)y&=&0 
 \end{array}
 \right.
 \end{eqnarray} 
To obtain a parametric polynomial system, we assign the intervals $[-1,2), [0,1), [3,5), [-3,1)$ and $[1,3)$ by $h_1,\ldots,h_5$ respectively. Then we observe the parametric polynomial system
$$S^* = \{h_1 x y+h_2 y+h_3, h_4 x y^2 + h_5y \} \subset \R[h_1,\ldots,h_5][x,y]$$
Using the lexicographic monomial ordering $y\prec x$ we can compute a comprehensive Gr\"obner system for $\langle S^* \rangle$ which contains about $19$ triples. However some of them are admissible only for some values of parameters {\it out} of their assigned interval. For instance the triple $(\{1\}, \{h_1,h_2,h_4,h_5\}, \{h_3\})$ is not acceptable in this example, since $h_5 \in [1,3]$ and so it can not be zero. By removing such triples, there remains only $8$ one shown in the following table. Therefore the following table shows $\mathcal{G}_\prec( [S])$.
\begin{table}[H]
 \centering {\small
 \begin{tabular}{|c||c||c|}
 \cline{1-3}
$E_i$&$N_i$&$G_i$\\
 \cline{1-3}
$\{h_3h_5\}$&$\{h_1, h_2, h_4 \}$& $\{1\}$\\
 \cline{1-3}
$\{h_3 h_4\}$&$\{h_1, h_2 \}$&$ \{1\}$\\
\cline{1-3}
 $\{h_1 h_2 h_3 h_4\}$&$\{h_3 h_4-h_1 h_5 \}$&$\{ 1\}$\\
\cline{1-3}
$\{h_2 h_3 h_5\}$&$\{h_1, h_4\}$&$\{1\}$\\
\cline{1-3}
$\{h_1 h_3 h_5\}$&$\{h_4 \}$& $\{1\}$\\
\cline{1-3}
 $\{h_2 h_3 h_4\}$&$\{h_1\}$&$\{h_2 y+h_3, h_3^2 h_4 x-h_2 h_5 h_3 \}$\\
\cline{1-3}
 $\{h_1 h_4\}$&$\{h_2, h_3 h_4-h_1 h_5 \}$&$\{h_3+h_1 x y \}$\\
\cline{1-3}
&&$\{h_1 x y+h_2 y+h_3, -h_4 h_2 y^2-h_4  h_3 y+h_5 h_1 y, $\\
$\{h_1 h_2 h_3 h_4( h_3 h_4-h_1 h_5)\}$ &$\{ \}$ & $-h_1 h_3^2 h_4 x+ h_1^2 h_5 h_3 x+h_2^2 h_4 h_3 y+h_2 h_4 h_3^2,$\\
 &&$-h_4  h_2 h_3 y -h_3^2 h_4+h_1 h_5 h_3\}$\\
 \cline{1-3}
 \end{tabular}
 }
\caption{Interval Gr\"obner system of System (5)} 
 \label{IGSexampletable}
 \end{table}
\end{exmp}

\subsection{Computing Interval Gr\"obner Systems}
\label{Sec3:IGS}
In tis section we state our algorithm so called {\sc IGS} to compute interval Gr\"obner system for an interval polynomial system. This algorithm is based on the {\sc PGB} algorithm with some extra conditions for the definition of consistency.
To begin let $ [S]=\{[f]_1,\ldots,[f]_\ell\}\subset [\mathbf{R}][x_1,\ldots,x_n]$ be a system of interval polynomials, where for each $1\leq j\leq \ell$, 
$$[f]_j=\sum_{i=1}^{m_j} [a_{ij},b_{ij}] x_1^{\alpha_{ij,1}}\cdots x_n^{\alpha_{ij,n}}$$
and $(\alpha_{ij,1},\ldots,\alpha_{ij,n})\in \mathbf{Z}^n_{\geq 0}$, for each $i$. 
As it is mentioned in Theorem 6.3, we assign to each interval coefficient $[a_{ij}, b_{ij}]$ a parameter $h_{ij}$ to convert $ [S]$ to a parametric polynomial system $S^*$.
The following proposition describes the relations between comprehensive Gr\"obner systems of $S^*$ and interval Gr\"obner bases of $ [S]$.
\begin{prop}
\label{Prop:1}
Using the above notations, let $[\mathcal{G}]$ and $\mathcal{G}$ be an interval Gr\"obner basis for $ [S]$ and a comprehensive Gr\"obner basis for $S^*$ respectively w.r.t. the same monomial ordering. Then for each $(E,N,G) \in [\mathcal{G}]$, there exists $(E',N',G') \in \mathcal{G}$ such that ${\bf V}(E)\setminus {\bf V}(N) \subset {\bf V}(E')\setminus {\bf V}(N')$ and $G,G'$ have the same initial ideal.
\end{prop} 
\begin{proof}
This comes from Definitions 6.4 and 5.5.
\end{proof}
According to the above proposition, to compute an interval Gr\"obner basis for $ [S]$, it is enough to compute a comprehensive Gr\"obner basis for $S^*$, and use a criterion to omit 
those triples $(E,N,G)$ lying in $\mathcal{G} \setminus [\mathcal{G}]$, we call them {\it redundant} triples.
\begin{rem}
\label{Rem:1}
Note that for a triple $(E,N,G)$ in $\mathcal{G} \setminus [\mathcal{G}]$, the intersection of ${\bf V}(E)\setminus {\bf V}(N) $ with the cartesian product of interval coefficients is empty.
\end{rem}
We are now going to present a criterion to determine the elements of $\mathcal{G} \setminus [\mathcal{G}]$. This criterion of course is based on the answer of this question:

~

\noindent{\it How can we sure that a system of polynomials $E\subset \mathbf{R}[a_1,\ldots,a_t]$ has a real root in the interval $[\alpha_1,\beta_1)\times \cdots\times [\alpha_t,\beta_t)$?}

~

In the case for which $\langle E \rangle$ is zero dimensional, this question is answered totally thanks to some efficient computational tools like Sturm's chain by isolating the real roots.
However in the case of positive dimensional, there exist only some algorithm to isolate the real roots.
Among them there exists an algorithm which determines whether a multivariate polynomial system has real root or not. 

Because of the reasons declared above, we convert the above key question to the problem of determining whether a polynomial system has a real root or not. 
\begin{thm}
\label{Thm:Crit}
Let  $E\subset \mathbf{R}[a_1,\ldots,a_t]$ be a finite polynomial set. Let also
$$F=E \cup \{a_i+(a_i-\beta_i)b_i^2-\alpha_i\ | \ i=1,\ldots,t\} \subset \mathbf{R}[a_1,\ldots,a_t,b_1,\ldots,b_t]$$
where $b_j$'s are algebraic independent by $a_i$'s and $[\alpha_i,\beta_i)$ be a real interval for each $i=1,\ldots,t$.
 Then the system $E=0$ has a solution in $[\alpha_1,\beta_1)\times \cdots\times [\alpha_t,\beta_t)$ if and only if 
the system $F=0$ has a real solution.
\end{thm}
\begin{proof}
Let $E=0$ has a solution $(\gamma_1,\ldots,\gamma_t)\in [\alpha_1,\beta_1)\times \cdots\times [\alpha_t,\beta_t)$. 
Let also
$$\eta_i=\sqrt{\frac{\alpha_i-\gamma_i}{\gamma_i-\beta_i}}$$
for each $i=1,\ldots,t$. It is easy to see that 
$$\gamma_i+(\gamma_i-\beta_i)\eta_i^2-\alpha_i = 0$$
which implies that 
$(\gamma_1,\ldots,\gamma_t,\eta_1,\ldots,\eta_t)$
is a solution of $F=0$.

Conversely, suppose that there exists $(\gamma_1,\ldots,\gamma_t,\eta_1,\ldots,\eta_t)\in \mathbf{R}^{2t}$ which is a solution of $F=0$, i. e. $f(\gamma_1,\ldots,\gamma_t)=0$ for each $f\in F$ and $\gamma_i+(\gamma_i-\beta_i)\eta_i^2-\alpha_i = 0$, for each $i=1\ldots,t$. It is enough to show that $\gamma_i \in [\alpha_i,\beta_i)$. In doing so, we see that 
$$\gamma_i = \frac{\alpha_i+\beta_i \eta_i^2}{1+\eta_i^2} = (\beta_i-\alpha_i)\frac{\eta_i^2}{1+\eta_i^2}+\alpha_i.$$
Indeed, $0\leq\frac{\eta_i^2}{1+\eta_i^2}<1$ and this shows that $\alpha_i\leq \underbrace{(\beta_i-\alpha_i)\frac{\eta_i^2}{1+\eta_i^2}+\alpha_i}_{\gamma_{_i}} <\beta_i$ which finishes the proof.

\end{proof}
\begin{rem}
\label{aux}
Note that for the intervals $[\alpha, \infty)$ and $(-\infty, \beta]$ we can use the auxiliary  polynomials $a-\alpha-b^2$ and $a-\beta+b^2$ respectively.
\end{rem}
Using Theorem 6.9 and the Remarks 6.8, 6.10 we can determine the elements of $\mathcal{G} \setminus [\mathcal{G}]$ exactly (see the notations of Proposition 6.7).
\begin{cor}
Let $(E,N,G) \in \mathcal{G}$ and $[\alpha_1,\beta_1),\ldots,[\alpha_t,\beta_t)$ be $t$ real intervals. Then 
$(E,N,G)$ is redundant if and only if the system $F=0$ has no real roots, where 
$$F=E \cup \{a_i+(a_i-\beta_i)b_i^2-\alpha_i\ | \ i=1,\ldots,t\} \cup \{\prod_{g\in N}(c_g g-1)\}$$
\noindent$ \subset \mathbf{R}[a_1,\ldots,a_t,b_1,\ldots,b_t,c_g: g\in N].$
\end{cor}
\begin{proof}
The proof comes from Theorem 6.9 and this fact that if $\prod_{g\in N}(c_g g-1)=0$ then there exists a $g\in N$ for which $c_g g-1 = 0$ which implies that $g\ne 0$.
\end{proof}
The above corollary is the criterion which determines all redundant triples, and so tackling this criterion with {\sc PGB} algorithm we can design our new algorithm to compute interval Gr\"obner systems. 
We design now the {\sc IGS} algorithm by its main procedure.
\begin{algorithm}[H]
\caption{{\sc IGS}}
\begin{algorithmic}
\Procedure{{\sc IGS}}{$S_{[\ ]}, \prec_{\bf x}$}
    \State {Assign $a_1,\ldots,a_t$ to interval coefficients and name it $S^*$};
    \State {$\prec_{\bf a}:=$ an arbitrary monomial ordering on $a_1,\ldots,a_t$;}
    \State {$E,N:=\{\ \},\{1\};$}
    \State {$\prec_{{\bf x}, {\bf a}}:=$The block ordering of $\prec_{\bf x}, \prec_{\bf a}$}
    \State {{\bf Return} {\sc PGB-main}$(P,E,N,\prec_{{\bf x}, {\bf a}}, L); \ \ \backslash\backslash L$ is the ordered set of interval coefficients which is needed to check consistency}
\EndProcedure
\end{algorithmic}
\end{algorithm}
The {\sc PGB-main} algorithm is the same which which is used in {\sc PGB} algorithm. We only change the definition of consistency as below. 
\begin{defn}
Let $[\alpha_1,\beta_1),\ldots,[\alpha_t,\beta_t)$ be $t$ real intervals and $E,N\subset \mathbf{R}[a_1,\ldots,a_t]$. The pair $(E,N)$ is called consistent if it is not redundant, or equivalently,
 $$[{\bf V}(E) \setminus {\bf V}(N)] \cap [\alpha_1,\beta_1)\times \cdots \times [\alpha_t,\beta_t)\neq \emptyset.$$
\end{defn} 
According to the above definition, we change the {\sc IsConsistent} algorithm to {\sc Interval-IsConsistent}, which checks the consistency for radical membership and redundancy determination.  

\begin{algorithm}[H]
\label{Alg:IsCo}
\caption{{\sc Interval-IsConsistent}}
\begin{algorithmic}
\Procedure{{\sc Interval-IsConsistent}}{$E,N, [\alpha_1,\beta_1),\ldots,[\alpha_t,\beta_t)$}
    \State {$test:=$false;}
    \State {$flag:=$false;}
    \If {$(E,N)$ is not redundant}
    	  \State {$test:=$true;}
    \EndIf
    \If{test}
	    \State {$flag:=$false;}
    	\For{$g\in N$ {\bf while} $flag=$false}
    		\If{$g\notin \sqrt{\langle E\rangle}$}
    			\State{$flag:=$true;}
    		\EndIf
    	\EndFor
    \EndIf
    \State {{\bf Return} $flag$};
\EndProcedure
\end{algorithmic}
\end{algorithm}
\begin{rem}
It is worth noting that redundant triples will be omitted before the algorithm goes to continue with them. This property causes that {\sc IGS} returns less triples than {\sc PGB}. 
\end{rem}

\section{Examples}
\label{Examples}
\subsection{Solving Fuzzy Polynomial Systems}
\label{Sec4:App}
In this section we state the ability of interval Gr\"obner system to solve parametric polynomial systems for which the parameters range over specific intervals, which may appear, for instance, when analysing fuzzy polynomial systems. 
In doing so, we can use the same {\sc IGS} algorithm, after converting an interval polynomial system to a parametric one. To give an example, we point at resolution of fuzzy polynomial systems which converts to solve a parametric polynomial system with parameters range over $[0,1]$ (see [1]). 

\begin{exmp}([1, Example 4.3] resolved)
Consider the following system of fuzzy polynomials:
$$
\left\{
\begin{array}{lll}
(0.25, 0.75, 0.625)&=& (0, 0.5, 0,5)x^5  + (0.5, 0.5, 0.25)y\\
(1.25, 1.25, 0.75)&=&(0.5,1,0.5)x^4+(1.5, 0.5, 0.5)xy\\
\end{array}
\right.
$$
To solve this system, the general way is to decompose  it  into four parametric polynomial systems. Here we solve only the third and fourth systems.
Consider the third system as follows:
$$
(3): \left\{
\begin{array}{lll}
f_1&=& (\frac{1}{2}-\frac{1}{2}h)x^5+(\frac{1}{2}h)y+\frac{1}{2}-\frac{3}{4}h\\
f_2&=& (-\frac{1}{2}+\frac{1}{2}h)x^5+(\frac{3}{4}-\frac{1}{4}h)y-\frac{7}{8}+\frac{5}{8}h\\
f_3&=& (-\frac{1}{2}+h)x^4+(2-\frac{1}{2}h)xy-\frac{5}{4}h\\
f_4&=& (1-\frac{1}{2}h)x^4+(1+\frac{1}{2}h)xy-2+\frac{3}{4}h\\
\end{array}
\right.
$$
 Computing an interval Gr\"obner basis w.r.t. $y\prec_{lex} x$, we have only a triple $(\{\ \}, \{h-1\}, \{1\})$ which means that we have no solutions whenever $h\ne 1$. As our algorithm deals with the values of $h\in [0,1)$, we must consider the case $h=1$ seperately. In this case, we find a solution $(h=1, x= -1.473157368, y=0.5)$ regarding the sign of variables.

The fourth system is:
$$
(4): \left\{
\begin{array}{lll}
f_1&=& (\frac{1}{2}-\frac{1}{2}h)x^5+(\frac{3}{4}-\frac{1}{4}h)y+\frac{1}{2}-\frac{3}{4}h\\
f_2&=& (-\frac{1}{2}+\frac{1}{2}h)x^5+(\frac{1}{2}h)y-\frac{7}{8}+\frac{5}{8}h\\
f_3&=& (-\frac{1}{2}+h)x^4+(1+\frac{1}{2}h)xy-\frac{5}{4}h\\
f_4&=& (1-\frac{1}{2}h)x^4+(2-\frac{1}{2}h)xy-2+\frac{3}{4}h\\
\end{array}
\right.
$$
where $h\in [0,1]$, $x\leq 0$ and $y\leq 0$. Computing an interval Gr\"obner basis, we see only one triple $(\{\ \}, \{h^2-2h+3\}, \{1\})$ for which its Gr\"obner basis equals to $\{1\}$. Therefore the above system has no solution for $h\in [0,1)$. Note that our method deals with $[0,1)$ here and so we must check the system for $h=1$ separately. By this, the system has no solution again by considering the sign of $x$ and $y$.  
and so this system has no solution.

Note that another way to inform that these system have or have not any real solutions (and not the exact form of solutions), is to check their consistency by what we have stated in Algorithm 6.1, since all of variables have interval form in these systems. 
In doing so we must add three auxiliary polynomials $h+(h-1)\tilde{h}^2, x+\tilde{x}^2$ and $y-\tilde{y}^2$ to the System (3) and 
$h+(h-1)\tilde{h}^2, x+\tilde{x}^2$ and $y+\tilde{y}^2$ to the System (4), where $\tilde{h}, \tilde{x}$ and $\tilde{y}$ are some extra real variables. The result of course is that (3) is consistent however (4) is not.
\end{exmp}
\subsection{Real Factors}
One of the interesting problems in the context of interval polynomials is the {\it Divisibility Problem}  stated as follows (See [32]):

{\it
For an interval polynomial ${[f]}\in [\mathbb{R}]\{x_1,\ldots,x_n\}$ and a real polynomial $g\in \mathbb{R}[x_1.\ldots,x_n]$, determine
whether there is a polynomial $p \in \mathcal{F}({[f]})$ such that $g$ is a factor of $p$.
}

In [32] there is stated an efficient method based on linear programming techniques and nice properties of polytopes. In the sequel we exalin our method to solve this problem by interval Gr\"obner basis.
To continue, let us say that $g$ i-divides $[f]$ whenever the answer of the above problem is {\it yes} (Note that the letter "i" stands for interval). By this conception we can now explain the following criterion based on interval Gr\"obner basis. 
\begin{thm}
\label{idiv}
Using the above notations, let $[S]=\{ [f], g\}\subset [\mathbb{R}]\{x_1,\ldots,x_n\}$ and $[\mathcal{G}]$ be a reduced interval Gr\"obner system for $[S]$ with respect to a monomial ordering $\prec$. Then $g$ i-divides $[f]$ if and only if there exists a triple $(E,N,G)\in [\mathcal{G}]$ where $G=\{g\}$ and ${\bf V}(E) \setminus {\bf V}(N) \ne \{0\}$.  
\end{thm}
\begin{proof}
It is easy to see that $g$ i-divides $[f]$ if and only if there exists an specialization $\sigma$ for which $g | \sigma([f])$ and of course $\sigma([f])\ne 0$, by the statement of divisibility problem. This implies that $\sigma([S])$ can be expressed only by $ \{ g \}$ and therefore there exists a pair of parametric sets $(E,N)$ such that $(E,N,\{g\}) \in [\mathcal{G}]$.
\end{proof}
\begin{exmp}
We are going to solve the divisibility problem for $$[f]=[-1,1]x^2+[-3,1]y^2+[1/2,2]xy\in [\mathbb{R}]\{x,y\}$$ and $g=x-\sqrt{2}y\in \mathbb{R}[x,y]$. By computing a reduced interval Gr\"obner basis for $[S]=\{ [f], g \}$, we find the triple 
$$(\{c\sqrt{2}+2a+b\}, \{1\}, \{x-\sqrt{2}y\})$$
where $a,b$ and $c$ denote inner values of the intervals $[-1,1]$, $[-3,1]$ and $[1/2,2]$. This means that if the values of $a,b$ and $c$ satisfy the equation $c\sqrt{2}+2a+b = 0$, then there exists some $p \in \mathcal{F}([f])$ such that $g | p$. For instance, by evaluating $a=1/3, b=-2$ and so $c=2 \sqrt{2}/3$ we find $p=1/3x^2-2y^2+2 \sqrt{2}/3 xy \in \mathcal{F}([f])$ which is divided by $g$ (Note that $p=1/3(x+3\sqrt{2}y) g$). Therefore $g$ i-divides $[f]$.
 
\end{exmp} 

\begin{rem}
We can use Theorem 7.2 for further aims. Let $f,g \in \mathbb{R}[x_1,\ldots,x_n]$ where $g$ does not divide $f$. Then one can use Theorem 7.2 to find a polynomial $\tilde{f}$ with the same coefficients of $f$ which contain a few perturbation and $g\mid \tilde{f}$. In doing so, one can convert $f$ to an interval polynomial $[f]$ by putting the interval $[c-\epsilon, c+\epsilon]$ instead of the coefficient $c$, for each coefficient $c$ appearing in $f$. Then using Theorem 7.2 one can increase the $\epsilon$ up enough until $g$ i-divides $[f]$ with the desired precision. 
\end{rem}

\section{Conclusion}
In the current paper we have introduced the concept of interval Gr\"obner system as a novel computational tool to analyse interval polynomial systems. We have further designed a complete algorithm to compute it using the existing methods to analyse parametric polynomial systems. This concept can solve some important problems from the applied and engineering research fields such as solving fuzzy polynomial system, as devoted in this paper.


\section*{Acknowledgement}
The authors would like to express their sincere thanks to Professor Deepak Kapur,  
who devoted his time and knowledge during the preparation of this paper. 
The third author would also like to give an special thanks to professor Prungchan Wongwises for her kindness on posting her Ph. D. thesis for him.

\bibliographystyle{plain}

\end{document}